\documentclass[conference,a4paper]{IEEEtran}
\usepackage{amssymb, amsthm,bm,mathtools} 
\usepackage{color}
\usepackage{comment}
\usepackage[utf8]{inputenc} 
\usepackage[T1]{fontenc}

\usepackage{url}             
\usepackage{cite} 
\usepackage[nameinlink,capitalize]{cleveref}
\usepackage[dvipsnames]{xcolor}

\usepackage{mleftright}       
\mleftright

\newtheorem{theorem}{Theorem}

\newtheorem{example}[theorem]{Example}

\newtheorem{definition}[theorem]{Definition}
\newtheorem{construction}[theorem]{Construction}
 
\usepackage{graphicx}         
\usepackage{siunitx, booktabs}   

\usepackage{tikz}
\usetikzlibrary{arrows,automata, positioning, calc, chains,
                fit, shapes, patterns}     
\usepackage{pgfplots}                        
\usepackage{subcaption}

\newcommand{\6}{\mathbf}
\newcommand{\wtH}{\textnormal{wt}_\textnormal{H}} 
\newcommand{\dH}{\textnormal{d}_\textnormal{H}} 
\newcommand{\wtsH}{\textnormal{wt}} 
\newcommand{\wtT}{\textnormal{wt}_\mathcal{T}}

\newcommand{\wH}{{weighted-Hamming} }
\newcommand{\WH}{{Weighted-Hamming} }

\newcommand{\ball}{\mathcal{B}}

\DeclareMathOperator*{\argmax}{arg\,max}
\DeclareMathOperator*{\argmin}{arg\,min}

\pgfplotsset{
legend image code/.code={
\draw[mark repeat=2,mark phase=2]
plot coordinates {
(0cm,0cm)
(0.2cm,0cm)    
(0.4cm,0cm)  
};
}
}
\pgfplotsset{compat=1.18}
\begin{document}

\title{{\WH Metric for
Parallel Channels}} 

\author{%
  \IEEEauthorblockN{Sebastian Bitzer$^1$, Alberto Ravagnani$^2$,  Violetta Weger$^1$}\\
  \IEEEauthorblockA{
                      $^1$Technical University of Munich, Germany \\ \{sebastian.bitzer, violetta.weger\}@tum.de\\
                      $^2$Eindhoven University of Technology, the Netherlands \\   a.ravagnani@tue.nl
                    }
}

\maketitle

\begin{abstract}
Independent parallel $q$-ary symmetric channels are a suitable transmission model for several applications.
The \wH metric is tailored to this setting and enables optimal decoding performance.
We show that some weighted-Hamming-metric codes exhibit the unusual property that all errors beyond half the minimum distance can be corrected.
Nevertheless, a tight relation between the error-correction capability of a code and its minimum distance can be established.
Generalizing their Hamming-metric counterparts, upper and lower bounds on the cardinality of a code with a given \wH distance are obtained.
Finally, we propose a simple code construction with optimal minimum distance for specific parameters.
\end{abstract}

\section{Introduction}\label{sec:intro}

For numerous practical applications, a suitable transmission model is given by parallel channels.
Examples include bit-interleaved coded modulation~\cite{caire1998bicm}, multi-carrier communications, but also, e.g., physical unclonable functions~\cite{sadr2013weighted, maringer2021VBSC}.
Standard coding solutions include repeat-accumulate \cite{liu2006reliable}, Raptor~\cite{mitzenmacher2006theory,richard2010design}, turbo~\cite{sason2007coding}, LDPC~\cite{ratzer2003sparse, sason2007achievableLDPC}, and polar codes~\cite{hof2012capacity}.
In practice, these probabilistic coding schemes can operate close to capacity.
However, they fail to give strong worst-case guarantees on the decoding failure rate.

Algebraic coding theory provides such guarantees for the setting of multiple highly correlated parallel channels \cite{gabidulin1973combinatorial, ahlswede2002parallel}. 
In particular, a series of publications considers burst error correction via joint decoding of interleaved codes \cite{krachkovsky1997decoding, schmidt2009collaborative}.
The case of \emph{non-interacting} parallel $q$-ary symmetric channels has received less attention.
Coding schemes include mismatched decoding of merged subchannels, independent coding for the subchannels and simple concatenated schemes~\cite{max1968parallel}.
These solutions are suboptimal in general. 

This work approaches the setting of independent parallel channels via the \emph{\wH metric}~\cite{deza2009encyclopedia}.
This generalization of the Hamming metric corresponds to the weighted sum of the $\mathcal{T}$-weight proposed in \cite{simonis}, which records the Hamming weights for a partition of the coordinates. 
Unlike for Forney's generalized distance measure \cite{forney1966generalized}, the optimal scaling factors assigned to the individual indices depend only on the transition probabilities of the subchannels, not the received sequence which does not provide further reliability information.
Particular constructions of codes endowed with the weighted-Hamming metric have been considered in~\cite{bezzateev2013class,moon2018weighted}, even though the subject is still in its beginnings, and many open questions remain.

The remainder of this paper is structured as follows.
Section~\ref{sec:new_weight} recalls the \wH metric; we analyze its connection to parallel $q$-ary symmetric channels and its properties regarding error correction.
By generalizing their Hamming-metric counterparts, Section~\ref{sec:new_bounds} derives upper and lower bounds on the minimum \wH distance of a code and, thus, its error-correction capability.
In Section~\ref{sec:new_codes}, we propose a construction of $q$-ary error-correcting codes that achieve optimal \wH distance for specific parameters.
Finally, we conclude the paper in Section~\ref{sec:conc}.

\section{The \WH Metric}\label{sec:new_weight}

We consider transmissions over $m$ parallel $q$-ary Symmetric Channels (QSC), each of which is characterized by a known crossover probability $\rho_\ell\in(0,1-\tfrac{1}{q})$ and acts independently of the other subchannels. 
Let $\6{c} = (\6c_1,\ldots,\6c_m)\in\mathbb{F}_q^{n}$ be the transmitted sequence, where $\6c_\ell\in\mathbb{F}_q^{n_\ell}$ is transmitted over the $\ell$-th QSC. 
Denote as $\wtH(\6c)=\left\lvert\{i\mid c_i \neq 0\}\right\rvert$ the Hamming weight of $\6c$.
The probability of receiving $\6{r} = (\6r_1,\ldots,\6r_m)$ is 
\[
P(\6r\,\vert\,\6c)= \prod_{\ell=1}^{m} \left(\frac{\rho_\ell}{q-1}\right)^{\wtH(\6r_\ell-\6c_\ell)} (1-\rho_\ell)^{n_\ell - \wtH(\6r_\ell-\6c_\ell)}.
\]
Here, unlike for a single QSC, the Hamming distance of $\6r$ and $\6 c$ does not suffice to fully characterize $P(\6r\,\vert\,\6c)$ since it cannot take into account the reliabilities implied by the individual crossover probabilities $\rho_\ell$.
This motivates the  use of the \wH metric \cite{deza2009encyclopedia}.
\begin{definition}[Weighted-Hamming metric]
Let $\6c = (\6c_1,\ldots,\6c_m)\in\mathbb{F}^n_q$ with $\6c_\ell\in\mathbb{F}_q^{n_\ell}$ and $n = \sum_{\ell=1}^{m}n_\ell$.
Fix scaling factors $\bm{\lambda} = (\lambda_1, \ldots, \lambda_m)\in\mathbb{N}^m$, the \emph{\wH weight} of $\6c$ is then
\[ \text{wt}(\6c)=
\sum_{\ell=1}^{m} \lambda_\ell \cdot\wtH(\6{c}_\ell).
\]
As usual, the \emph{\wH distance} between $\6c$ and $\6c' \in \mathbb{F}_q^n$ is given by $d(\6c,\6c')= \wtsH(\6c-\6c')$, and the minimum \wH distance of a code $\mathcal{C} \subseteq \mathbb{F}_q^n$ is defined as 
\[
d(\mathcal{C}) = \min \{ d(\6c, \6c') \mid \6c, \6c' \in \mathcal{C}, \6c \neq \6c'\}.
\]
\end{definition}

As we already mentioned in the introduction, the \wH metric can be seen as a weighted summation of the $\mathcal{T}$-weight defined in \cite{simonis}.  
Here, we consider only the case where the partition is chosen such that the subsets correspond to the subchannels of the transmission model.
The $\mathcal{T}$-weight of $\6c = (\6c_1,\ldots,\6c_m)$ is then given by the $m$-tuple
\[
\wtT(\6c) = \left(\wtH(\6c_1), \ldots, \wtH(\6c_m) \right)\in \mathbb{N}^m.
\]
The following theorem confirms that the \wH metric with suitable scalars $\bm{\lambda}\in\mathbb{N}^m$ indeed contains sufficient information for optimal decoding. 
\begin{theorem}\label{thm:ml}
Let $\6{r} = (\6r_1,\ldots,\6r_m)$ with $\6r_\ell\in\mathbb{F}_q^{n_\ell}$ be a sequence obtained by transmitting $\6{c} = (\6c_1,\ldots,\6c_m)\in\mathcal{C}$ over $m$ parallel $q$-ary symmetric channels, each with individual error probability $\rho_\ell\in\big(0,1-\tfrac{1}{q}\big)$.
Then, there exist $\lambda_\ell\in\mathbb{N}$, for $\ell\in\{1,\ldots,m\}$, such that maximum-likelihood decoding is obtained by minimizing the \wH distance between $\6r$ and a codeword $\6c \in \mathcal{C}$.
\end{theorem}
\begin{proof}
We generalize the classical argument for a single QSC, see, e.g., \cite{bossert1999channel}.
The maximum-likelihood estimate is defined as
\[
\hat{\6c} = \argmax_{(\6c_1,\ldots,\6c_m)\in\mathcal{C}} P\mleft((\6r_1,\ldots,\6r_m)\mid(\6c_1,\ldots,\6c_m)\mright).
\]
Since the QSCs are memoryless and non-interacting, we have
\[
\hat{\6c} = \argmax_{(\6c_1,\ldots,\6c_m)\in\mathcal{C}} \sum_{\ell=1}^{m} \tau_\ell\log\mleft(\frac{\rho_\ell}{q-1}\mright) + (n_\ell-\tau_\ell)\log\mleft(1-\rho_\ell\mright), 
\]
where $\rho_\ell$ denotes the error probability of the $\ell$-th QSC and $\tau_\ell = \dH(\6{r}_\ell, \6{c}_\ell)$. 
Since $n_\ell$ is independent of $\6{c}_\ell$, we obtain
\[
\hat{\6{c}} = \argmin_{(\6c_1,\ldots,\6c_m)\in\mathcal{C}} \sum_{\ell=1}^{m} \tau_\ell \lambda'_\ell,
\]
with $\lambda'_\ell = \log\big( \tfrac{1-\rho_\ell}{\rho_\ell}\big) + \log( q-1 ) > 0$ due to $\rho_\ell < 1-\tfrac{1}{q}$. 
Thus, using $(\lambda_1',\ldots,\lambda'_m) \in \mathbb{R}^m$ is equivalent to using $(\lambda_1,\ldots,\lambda_m) \in \mathbb{N}^m$ for $\lambda_\ell = \alpha\lambda_\ell'$ and a suitable $\alpha\in\mathbb{R}$.
\end{proof}

Due to \Cref{thm:ml}, a code that corrects all error patterns of \wH weight at most $t$ for suitable $\bm{\lambda}$ provides a guarantee on the resulting word error rate.
The following toy example illustrates the advantage of using the \wH metric compared to mismatched decoding in the Hamming metric, i.e., neglecting differing error probabilities of the subchannels. 
See Section~\ref{sec:new_codes} for a more general construction.

\begin{example}\label{ex:toy}
    Let $q=2$, $m=2$, $\bm{\lambda} = (1,2)$, and $n_1 = n_2 = 4$.
    Let $\mathcal{C}$ be spanned by $\6G = (\6I_4, \, \mathbf{1}-\6I_4),$ where $\6I$ denotes the identity matrix and $\mathbf{1}$ denotes the all-one matrix.
    $\mathcal{C}$ has dimension $4$ and minimum \wH distance $5$, which is optimal according to the bounds derived in Section \ref{sec:new_bounds}.
    For $\rho_1 = 0.125$ and $\rho_2 = 0.02$, $\mathcal{C}$ can correct all error patterns that occur with probability at least $0.011$.
    The largest code, which achieves the same performance via mismatched decoding in the Hamming metric, is the $2$-dimensional Cordaro-Wagner code \cite{cordaro1967optimum}.
    Independent coding for both subchannels only allows for transmitting a single bit using a length-$4$ repetition code in the second subchannel.
\end{example}

In \Cref{ex:toy}, we used that errors up to half the minimum distance are correctable.
In general, a linear code $\mathcal{C}$ can correct any pattern of 
weight $t$ if and only if $t\leq\tau(\mathcal{C})$, where
\[
\tau(\mathcal{C}) = \min_{\6c\in\mathcal{C}\setminus\{\60\}, \, \6r\in\mathbb{F}_q^n} \max\{\wtsH(\6r), \, \wtsH(\6c-\6r)\} - 1.
\]
One can show that $\tau(\mathcal{C})\geq \left\lfloor (d(\mathcal{C})-1)/{2}\right\rfloor$, with equality for \emph{normal} discrepancy functions\cite{silva2009metrics}. 
It is well-known that this is the case for common metrics, such as the Hamming or the rank metric.
The \wH distance is, however, not normal and there exist codes that can correct all error patterns with weight beyond half the minimum distance, as the following example shows.

\begin{example}\label{ex}
Let $m=2$, $n_1 = n_2 = 4$, $k=4$, $\bm{\lambda} = (2,7)$. 
We consider the code generated by $\mathbf{G} = (\mathbf{0},\, \mathbf{I}_4),$ which has minimum \wH distance $d(\mathcal{C}) = 7$. 
Despite having $\lfloor (d(\mathcal{C})-1)/2\rfloor=3$, this code can correct all error patterns of weight $t \leq 6 = \tau(\mathcal{C})$.
\end{example}

This is in contrast to the Hamming metric, where some errors of weight $t$ larger than $\lfloor (d-1)/2\rfloor$ can be corrected, but not all.
Consequently, the minimum distance might underestimate a code's true guaranteed error-correction capability.
As the following theorem shows, we can still bound the  error-correction capability by means of the minimum distance.

\begin{theorem}
Let $\lambda_1 \leq \ldots \leq \lambda_m.$ Then, it holds that
\[ \left\lfloor \frac{d(\mathcal{C})-1}{2}\right\rfloor \leq
\tau(\mathcal{C})\leq \left\lfloor\frac{d(\mathcal{C})+\lambda_m}{2}\right\rfloor -1.
\]
\end{theorem}
\begin{proof}
The first inequality holds for arbitrary metrics \cite{silva2009metrics}.
To prove the second inequality, we bound the error-correction capability as 
\[\tau(\mathcal{C}) \leq \min_{\6r\in\mathbb{F}_q^n} \max\{\text{wt}(\6r), \text{wt}(\6c-\6r)\} - 1,
\]
where we choose $\6c$ as a minimum \wH weight codeword, i.e., $\text{wt}(\6c) = d(\mathcal{C})$.
Denote as $\{i_1,\ldots,i_{w}\} \subseteq \{1, \ldots, n\}$ the set of nonzero coordinates of $\6c$.
We construct an explicit $\6r\in\mathbb{F}_q^n$ as
$$
r_i = 
\begin{cases}
c_i &\text{for }i \in\{i_{2j}\mid j=1,\ldots,\lfloor w/2\rfloor\},\\
0&\text{else.}\\
\end{cases}
$$
Let $w_\ell = \wtH(\6 c_\ell)$.
For blocks with even $w_\ell$, this construction assigns $\tfrac{w_\ell}{2}$ elements to $\6r_\ell$ and to $(\6c - \6r)_\ell$.
For blocks with odd $w_\ell$, one is assigned $\tfrac{w_\ell-1}{2}$ and the other $\tfrac{w_\ell + 1}{2}$ in alternating order.
Denote as $\{\ell_1, \ldots, \ell_{s}\} \subseteq \{1, \ldots, m\}$ the set of blocks with $w_\ell$ odd.
Then, 
\[
\max\{\text{wt}(\6r), \text{wt}(\6c-\6r)\} = \sum_{\ell = 1}^m \wtH(\6c_{\ell}) \tfrac{\lambda_\ell}{2} + \left\lvert\sum_{j=1}^{s} (-1)^{j}\tfrac{\lambda_{\ell_j}}{2} \right\rvert.
\]
Due to $\lambda_1\leq\ldots\leq \lambda_m$, one can bound the last expression as
\[
\left\lvert\sum_{j=1}^{s} (-1)^{j}\tfrac{\lambda_{\ell_j}}{2} \right\rvert\leq \tfrac{\lambda_{\ell_s}}{2} \leq \tfrac{\lambda_m}{2}
\]
and the statement follows by rounding.
\end{proof}

Note that the upper and lower bound on the error correction capability are tight, e.g., the code in Example~\ref{ex} attains the upper bound.
Since the actual error-correction capability of a code heavily depends on its structure, in the following coding-theoretic bounds, we consider the minimum distance as the main code parameter.
Nevertheless, we want to point out the opportunity for further research regarding bounds using $\tau(\mathcal{C})$.

\section{Bounds}\label{sec:new_bounds}

This section explores the generalization of known bounds for the Hamming metric to the \wH metric.
Let us start with a Singleton-like bound, which we derive using the anticode argument (see, e.g., \cite{anticode}).
\begin{theorem}[Singleton-like bound]\label{theo:singleton}
Let $\mathcal{C} \subseteq \mathbb{F}_{q}^n$ be a code with minimum distance $d$. 
Assume $\lambda_1 \leq \ldots \leq \lambda_{m}$ and let $\ell^* \in \{0,\ldots,m-1\}$ be the largest {s.t.} $\sum_{\ell=1}^{\ell^*} n_\ell \lambda_\ell < d$. 
Then,
\[
\log_q(\left\lvert \mathcal{C} \right\rvert) \leq \sum_{\ell=\ell^*+1}^m n_\ell - \left\lfloor \frac{d-1-\sum_{\ell=1}^{\ell^*} n_\ell \lambda_\ell}{\lambda_{\ell^* +1}}\right\rfloor.
\]
\end{theorem}
\begin{proof}
Since a code $\mathcal{C}' \subseteq \mathbb{F}_q^n$ with maximal \wH weight $<d$ has to be such that $\lvert\mathcal{C}\rvert\lvert\mathcal{C}'\rvert \leq q^n$ (as these codes can only intersect trivially), we can equivalently consider a lower bound on $\log_q(\lvert\mathcal{C}'\rvert)$.
Let $\ell^* \in \{1, \ldots, m-1\}$ be the largest integer such that $\sum_{\ell=1}^{\ell^*} n_\ell \lambda_\ell < d$ and define 
\[
x = \left\lfloor\frac{d-1-\sum_{\ell=1}^{\ell^*} n_\ell \lambda_\ell}{\lambda_{\ell^*+1}}\right\rfloor \in\{0,\ldots,n_{\ell^*+1}\}.
\]
We set $k'= \sum_{i=1}^{\ell^*} n_i + x$ and choose the code $\mathcal{C}'$ generated by $\6G'=\begin{pmatrix}\mathbf{I}_{k'} & \mathbf{0} \end{pmatrix} \in \mathbb{F}_q^{k'\times n}$.
Clearly, $\mathcal{C}'$ has maximum \wH distance less than $d$.
Due to $k\leq n -k'$, we have
\begin{equation*}
k \leq n - \sum_{\ell=1}^{\ell^*} n_\ell - x =  \sum_{i=\ell^*+1}^{m} n_i - \left\lfloor\frac{d-1-\sum_{\ell=1}^{\ell^*} n_\ell \lambda_\ell}{\lambda_{\ell^*+1}}\right\rfloor. \qedhere
\end{equation*}
\end{proof}

The bound given in Theorem \ref{theo:singleton} is tight in the sense that optimal codes exist.
A subset of these codes is given by Maximum Distance Separable (MDS) codes, which are the optimal codes for the Hamming-metric Singleton bound, as the following shows.

\begin{theorem}
Assume $\lambda_1\leq\ldots\leq\lambda_m$ and let $\mathcal{C} \subseteq \mathbb{F}_q^n$ be an MDS code of dimension $k$.
Then, $\mathcal{C}$ has minimum \wH distance
\[
d = \sum_{\ell = 1}^{\ell'} n_\ell \lambda_\ell +  \left(n-k+1-\sum_{\ell=1}^{\ell'} n_\ell\right) \cdot\lambda_{\ell'+1},
\]
where $\ell'\in\{0,\ldots,m\}$ is maximal with $\sum_{\ell=1}^{\ell'} n_\ell \leq n-k+1$.
This implies that $\mathcal{C}$ is Maximum \WH Distance (MWHD), i.e., it attains the Singleton-like bound. 
\end{theorem}
\begin{proof}
Note that any $k$ positions of the MDS code form an information set.
Hence, there is a codeword $\6c$ with support $\{0,\ldots, n-k\}$.
Due to $\lambda_1\leq\ldots\leq\lambda_m$, this codeword has the smallest non-zero \wH weight among all  codewords of $\mathcal{C}$.
All elements of the first $\ell'$ blocks are contained in the support of $\6c$; the remaining $n-k+1-\sum_{\ell=1}^{\ell'} n_\ell$ non-zero entries are in block $\ell'+1$.
This gives the weight of $\6c$ and, therefore, the minimum distance of $\mathcal{C}$.

To compute the maximum dimension of a code with  minimum distance given as in Theorem~\ref{theo:singleton}, we need the largest $\ell^*$ such that
\[
\sum_{\ell=1}^{\ell^*} n_\ell\lambda_\ell < d = \sum_{\ell = 1}^{\ell'} n_\ell \lambda_\ell +  \left(n-k+1-\sum_{\ell=1}^{\ell'} n_\ell\right) \lambda_{\ell'+1}.
\]
Due to the definition of $\ell'$, $n-k+1-\sum_{\ell=1}^{\ell'} n_\ell < n_{\ell'+1}$ holds.
Consequently, $\ell' = \ell^*$ and the expression 
\[
\sum_{\ell=\ell^*+1}^m n_\ell - \left\lfloor \frac{d-1-\sum_{\ell=1}^{\ell^*} n_\ell \lambda_\ell}{\lambda_{\ell^* +1}}\right\rfloor
\]
simplifies to $k$. 
Hence, $\mathcal{C}$ is MWHD.
\end{proof}
It is well known that MDS codes are dense for $q$ going to infinity; thus, MWHD codes are also dense in this setting. 
To derive further density results as well as sphere-packing/covering bounds, the size of the \wH balls 
$\ball_q(n,r, \bm{\lambda}) = \{\6x \in \mathbb{F}_{q}^n \mid \wtsH(\6x) \leq r\}$ is required. 
Let
\begin{equation} \label{Ls}
\Lambda(s) = \left\{ (w_1,\ldots,w_m) \mathrel{\Big|}  \sum_{\ell=1}^m w_\ell \lambda_\ell=s, 0\leq w_\ell \leq n_\ell\right\}
\end{equation}
denote the set of $\mathcal{T}$-weights that correspond to the \wH weight $s$.
Then, the ball size is given by 
\[
\left\lvert \ball_q(n,r, \bm{\lambda}) \right\rvert 
= \sum_{s=0}^r \sum_{\6w\in\Lambda(s)} \prod_{\ell=1}^m \binom{n_\ell}{w_\ell} (q-1)^{\sum_{\ell=1}^m w_\ell},
\]
which can be computed efficiently using dynamic programming in a similar way as in \cite{puchinger2022generic}.
The sphere-packing and sphere-covering bounds follow from well-known arguments.
\begin{theorem}[Sphere-packing and sphere-covering bounds]\label{theo:Hamming}
    Denote by $A_q(n,d, \bm{\lambda})$ the largest size of a code in $\mathbb{F}_{q}^n$ of minimum distance $d$. 
    The \wH sphere-packing/Hamming bound~\cite{bezzateev2013class} states that
    \begin{align*} 
    A_q(n,d,\bm{\lambda}) & \leq \frac{q^{n}}{\lvert \ball_q(n,\lfloor \frac{d-1}{2} \rfloor, \bm{\lambda}) \rvert }.
    \end{align*}
    The sphere-covering/Gilbert-Varshamov bound states that
    \begin{align*} 
    A_q(n,d,\bm{\lambda}) & \geq \frac{q^{n}}{\lvert \ball_q(n,d-1, \bm{\lambda})\rvert }.
    \end{align*}
\end{theorem}

Let us now examine the asymptotic behavior of the Gilbert-Varshamov bound. 
We consider the setting where  $n_\ell = \lfloor \alpha_\ell \cdot n \rceil$ with fixed $\alpha_\ell$ {s.t.} $\sum_{\ell=1}^m \alpha_\ell = 1$.
The \wH weight of a vector $\6r \in \mathbb{F}_q^{n}$ is at most $M = \sum_{\ell = 1}^{m} \lambda_\ell n_\ell$.
Let us denote the relative  minimum \wH  distance as $\delta = d/M$, the asymptotic size of the balls as 
\[
g_q(\delta, \bm{\lambda})= \lim_{n \to \infty} \frac{1}{n} \log_q( \lvert \ball_q(n, \delta M, \bm{\lambda})\rvert),
\]
and the maximal information rate as  
\[
R_q(n,d, \bm{\lambda}) = \frac{1}{n} \log_q(A_q(n,d, \bm{\lambda})).
\]
Then, the asymptotic Gilbert-Varshamov bound states that
\[
\liminf\limits_{n \to \infty} R_q(n,\delta M, \bm{\lambda}) \geq 1- g_q(\delta, \bm{\lambda}).
\]

Let $D$ be the minimal value in $[0,1]$, such that $g_q(D, \bm{\lambda})=1.$
We can now show that random codes $\mathcal{C} \subseteq \mathbb{F}_q^n$ attain the Gilbert-Varshamov bound with high probability. 
The proof follows also directly from \cite[Theorem 20]{free}.
\begin{theorem}
For arbitrary $\delta \in [0,D),$  arbitrary $0 < \varepsilon < 1-g_q(\delta, \bm{\lambda})$ and sufficiently large $n$, the following holds for $k = \lceil (1-g_q(\delta, \bm{\lambda})-\varepsilon)n\rceil.$
If $\6G \in \mathbb{F}_q^{k \times n}$ is chosen uniformly at random, then the code generated by $\6G$ has rate $k/n$ and relative minimum  \wH distance $\delta$ with probability at least
\[
1-q^{(1-\varepsilon n)} \geq 1- e^{-\Omega(n)}.
\]  
\end{theorem}
\begin{proof}
The probability that $\6G \in \mathbb{F}_q^{k \times n}$ generates a code of dimension $k$, is well known to be $\prod_{i=0}^{k-1} (1-q^{i-n}),$ which tends to one as $n$ grows, and is in particular larger than $1-e^{-\Omega(n)}.$
Note that for any non-zero $\6x \in \mathbb{F}_q^k$, we have that $\6x\6G$ is uniform at random in $\mathbb{F}_q^n$ and thus, the probability that $\text{wt}(\6x\6G) \leq \delta M$ is at most 
\[
\frac{\lvert \ball_q(n, \delta M, \bm{\lambda}) \rvert}{q^n} \leq q^{n g_q(\delta, \bm{\lambda})-n}.
\]
Using a union bound over all non-zero $\6x \in \mathbb{F}_q^k$, we get that the probability for the code to have minimum distance $\delta M$ is bounded from above by 
$q^{k} q^{n g_q(\delta, \bm{\lambda})-n} \leq q^{1-\varepsilon \cdot n}.$ 
Thus, the code has relative minimum  \wH distance $\delta$ with probability at least
$1- q^{1-\varepsilon n} \geq 1- e^{-\Omega(n)}$.
\end{proof}

We continue by providing a Plotkin-like bound for linear codes endowed with the \wH metric. 
Recall that the maximal weight is given by $M= \sum_{\ell=1}^m n_\ell \lambda_\ell.$ 
\begin{theorem}[Plotkin-like bound]\label{theo:plot}
    Let $\mathcal{C} \subseteq \mathbb{F}_{q}^n$ be a linear code with minimum \wH distance $d$. 
    If $\smash{d > \big(\frac{q-1}{q}\big) M}$, then
    $$d \leq \frac{ \lvert \mathcal{C} \rvert}{\lvert \mathcal{C} \rvert -1} \left(\frac{q-1}{q}\right) M.$$
     \end{theorem}
\begin{proof}
This follows directly from the fact that the average  \wH weight on $\mathbb{F}_{q}^n$ is
\[ 
\bar{d}\coloneqq  \sum_{\ell=1}^m \frac{q-1}{q}{n_\ell} \lambda_{\ell}= \left(\frac{q-1}{q}\right) M,
\]
and that the average weight of a code $\mathcal{C}$ is 
\[
\overline{\text{wt}}(\mathcal{C})= \frac{1}{\lvert \mathcal{C} \rvert} \sum_{\6c \in \mathcal{C}} \text{wt}(\6c) \leq \bar{d}.
\]
The classical Plotkin argument gives $d \leq \frac{\lvert \mathcal{C} \rvert}{\lvert \mathcal{C} \rvert -1} \overline{\text{wt}}(\mathcal{C})$. 
For further details, the reader is referred to \cite{plotkin}.
\end{proof}

We now turn to a Linear Programming (LP) bound for the \wH metric.
As usual, the weight enumerator of a code $\mathcal{C}$ with respect to the \wH metric 
is defined as 
$A_i(\mathcal{C})= \left\lvert \{ \6c \in \mathcal{C} \mid \text{wt}(\6c)=i \} \right\rvert$ for $0\leq i \leq M$ and $M=\sum_{\ell=1}^{m} n_\ell \lambda_\ell$. 
However, to obtain an LP bound, we use the finer partition due to the $\mathcal{T}$-weight enumerator
\[
A_{\6i}(\mathcal{C})= \left\lvert \big\{ \6c \in \mathcal{C} \mid \text{wt}_\textnormal{H}(\6c_\ell)=i_\ell \text{ for all }  \ell \in \{1, \ldots,m\} \big\} \right\rvert
\]
for all $\6i \in \Lambda,$ where $\Lambda$ is the set of all possible $\mathcal{T}$-weights 
\[
\Lambda = \bigcup_{i=0}^M \Lambda(i)= \big\{ \6i \in \mathbb{N}^m \mid i_\ell \leq n_\ell \text{ for all } \ell \in \{1, \ldots, m\}\big\}.
\]
The two weight enumerators satisfy $A_i(\mathcal{C})= \sum_{\6i \in \Lambda(i)} A_{\6i}(\mathcal{C})$. 
Recall the $\mathcal{T}$-weight MacWilliams identities, as stated in \cite{simonis}.
\begin{theorem}
    Let $\mathcal{C} \subseteq \mathbb{F}_q^n$ be a linear code, and $\mathcal{C}^\perp$ denote its dual with respect to the standard inner product $\langle \cdot, \cdot \rangle$. We have 
    \[
    A_{\6j}(\mathcal{C}^\perp)= \frac{1}{\lvert \mathcal{C} \rvert} \sum_{\6i \in \Lambda} \prod_{\ell=1}^m K_{j_\ell}^{\textnormal{H}}(i_\ell) A_{\6i}(\mathcal{C}),
    \]
    where $K_{j_\ell}^{\textnormal{H}}(i_\ell)$ denotes the Hamming-metric Krawtchouk coefficient for the $\ell$-th block, which is given as
    \[
    K_{j_\ell}^\textnormal{H}(i_\ell)=\sum_{s=0}^{j_\ell} \binom{n_\ell-i_\ell}{{j_\ell}-s}\binom{i_\ell}{s}(q-1)^{{j_\ell}-s}(-1)^s.
    \]
\end{theorem}

We can use the MacWilliams identities of the $\mathcal{T}$-weight to obtain a linear programming bound for the \wH metric.
That is, we maximize $\sum_{\6i \in \Lambda} A_{\6i}$ under the linear constraints 
\begin{equation*}
\begin{aligned}
A_{\6{0}} &=1, \\
A_{\6i} &\geq 0 & \forall \6i &\in \Lambda,  \\
A_{\6i} &=0 & \forall \6i &\in \bigcup\nolimits_{w=1}^{d-1}\Lambda(w), \text{ and}\\
\textstyle\sum\nolimits_{\6i \in \Lambda} \textstyle\prod\nolimits_{\ell=1}^m    K_{j_\ell}^{\textnormal{H}}(i_\ell)A_{\6i} &\geq 0 & \forall \6j &\in \Lambda. 
\end{aligned}
\end{equation*}
The second and fourth condition ensure that the $\mathcal{T}$-weight enumerators of the code and the dual code are non-negative, while the third condition ensures that the code has minimum \wH distance at least $d$. Clearly $A_{\6i}=A_{\6i}(\mathcal{C})$ is a solution with $\sum_{\6i \in \Lambda} A_{\6i}(\mathcal{C})= \lvert \mathcal{C} \rvert.$

In Figure~\ref{fig:more_comparison}, we present a comparison of the provided bounds. 
We give parameters for which each of the bounds, namely  Singleton-like (Thm.~\ref{theo:singleton}),  Plotkin-like (Thm.~\ref{theo:plot}) and  Hamming-like (Thm.~\ref{theo:Hamming}), outperform the others. 
We fix the parameters $\6n=(n_1,n_2)=(7,7)$ and $\bm\lambda=(\lambda_1, \lambda_2)=(1,2)$.
In particular, for $q=2$, we have that for $d \in \{3,\ldots, 11\}$ the Hamming-like bound gives the tightest bound, whereas for $d \geq 12$ the Plotkin-like bound is superior. 
For $q=7$, we have that for $d \leq 19$ the Singleton-like bound outperforms the others, and again, as soon as the condition of the Plotkin-like bound is satisfied it provides the tightest bound.
The LP bound is always at least as tight as the others and improves them for particular parameters. 
Finally, also the  Gilbert-Varshamov (GV) bound, and an explicit code construction are included.
The latter is introduced as Construction \ref{const:simple} in the following section and is optimal for minimum distance $d=5$.

\begin{figure}[t]
\centering%
\begin{subfigure}[b]{0.5\columnwidth}
    \centering
\begin{tikzpicture}[baseline] 
\begin{axis}[
ymax = 14,
ymin = 0,
xmin = 1,
xmax = 21,
label style={font=\footnotesize},
ymajorgrids,
xmajorgrids,
grid style=dashed,
legend style={at={(1,1)},anchor=north east,font=\scriptsize,legend cell align=left, align=left, draw=black},
ylabel near ticks,
xlabel near ticks,
width = 5.5cm,
legend columns=6,
legend entries={Hamming\ , Singleton\ , Plotkin\ , LP\ \ , \ {Const.} \ref{const:simple}\ , GV},
legend to name={mylegend},
every axis y label/.style={font=\footnotesize,at={(current axis.north west)},above =0cm,left = 0.5mm},
ylabel={$k$},
ticklabel style = {font=\tiny},
]

\addplot [line width=1pt,Green] 
  table[row sep=crcr]{%
1 14 \\
2 14 \\
3 11 \\
4 11 \\
5 8 \\
6 8 \\
7 7 \\
8 7 \\
9 5 \\
10 5 \\
11 4 \\
12 4 \\
13 3 \\
14 3 \\
15 2 \\
16 2 \\
17 1 \\
18 1 \\
19 1 \\
20 1 \\
21 1 \\
};

\addplot [line width=1pt,CornflowerBlue]
  table[row sep=crcr]{%
1 14 \\
2 13 \\
3 12 \\
4 11 \\
5 10 \\
6 9 \\
7 8 \\
8 7 \\
9 7 \\
10 6 \\
11 6 \\
12 5 \\
13 5 \\
14 4 \\
15 4 \\
16 3 \\
17 3 \\
18 2 \\
19 2 \\
20 1 \\
21 1 \\
};

\addplot [line width=1pt,YellowOrange]
  table[row sep=crcr]{%
11 4 \\
12 3 \\
13 2 \\
14 2 \\
15 1 \\
16 1 \\
17 1 \\
18 1 \\
19 1 \\
20 1 \\
21 1 \\
};

\addplot [line width=1pt,black, densely dotted]
  table[row sep=crcr]{%
1 14 \\
2 13 \\
3 11 \\
4 10 \\
5 8 \\
6 8 \\
7 7 \\
8 6 \\
9 5 \\
10 4 \\
11 3 \\
12 3 \\
13 2 \\
14 2 \\
15 1 \\
16 1 \\
17 1 \\
18 1 \\
19 1 \\
20 1 \\
21 1 \\
};

\addplot[red, mark = *, only marks]
  table[row sep=crcr]{%
5 8 \\
};

\addplot [line width=1pt,gray, dashed]
  table[row sep=crcr]{%
1 14 \\
2 11 \\
3 9 \\
4 8 \\
5 6 \\
6 5 \\
7 4 \\
8 3 \\
9 2 \\
10 2 \\
11 1 \\
12 1 \\
13 1 \\
14 1 \\
15 1 \\
16 1 \\
17 1 \\
18 1 \\
19 1 \\
20 1 \\
21 1 \\
};

\end{axis}
\end{tikzpicture}
\caption{$q=2$}
\label{fig:code_a}
\end{subfigure}
\begin{subfigure}[b]{0.48\columnwidth}
\centering
\begin{tikzpicture}[baseline] 
\begin{axis}[%
ymax = 14,
ymin = 0,
xmin = 1,
xmax = 21,
label style={font=\footnotesize},
xlabel={$d$},
ymajorgrids,
xmajorgrids,
grid style=dashed,
legend style={at={(1,1)},anchor=north east,font=\scriptsize,legend cell align=left, align=left, draw=black},
xlabel near ticks,
width = 5.5cm,
ymajorticks=false,
ticklabel style = {font=\tiny},
every axis x label/.style={font=\footnotesize,at={(current axis.south east)},below =0cm,right = 0mm},
]

\addplot [line width=1pt,Green] 
  table[row sep=crcr]{%
1 14 \\
2 14 \\
3 12 \\
4 12 \\
5 10 \\
6 10 \\
7 9 \\
8 9 \\
9 8 \\
10 8 \\
11 7 \\
12 7 \\
13 6 \\
14 6 \\
15 5 \\
16 5 \\
17 4 \\
18 4 \\
19 4 \\
20 4 \\
21 3 \\
};

\addplot [line width=1pt,CornflowerBlue]
  table[row sep=crcr]{%
1 14.0000000000000 \\
2 13.0000000000000 \\
3 12.0000000000000 \\
4 11.0000000000000 \\
5 10.0000000000000 \\
6 9.00000000000000 \\
7 8.00000000000000 \\
8 7.00000000000000 \\
9 7.00000000000000 \\
10 6.00000000000000 \\
11 6.00000000000000 \\
12 5.00000000000000 \\
13 5.00000000000000 \\
14 4.00000000000000 \\
15 4.00000000000000 \\
16 3.00000000000000 \\
17 3.00000000000000 \\
18 2.00000000000000 \\
19 2.00000000000000 \\
20 1.00000000000000 \\
21 1.00000000000000 \\
};

\addplot [line width=1pt,black, densely dotted]
  table[row sep=crcr]{%
1 14 \\
2 13 \\
3 12 \\
4 11 \\
5 10 \\
6 9 \\
7 8 \\
8 7 \\
9 7 \\
10 6 \\
11 6 \\
12 5 \\
13 5 \\
14 4 \\
15 3 \\
16 3 \\
17 2 \\
18 2 \\
19 1 \\
20 1 \\
21 1 \\
};

\addplot [line width=1pt, YellowOrange]
  table[row sep=crcr]{%
19 1 \\
20 1 \\
21 1 \\
};

\addplot[red, mark = *, only marks]
  table[row sep=crcr]{%
5 10 \\
};

\addplot [line width=1pt,gray, dashed]
  table[row sep=crcr]{%
1 14 \\
2 13 \\
3 11 \\
4 10 \\
5 9 \\
6 8 \\
7 7 \\
8 6 \\
9 5 \\
10 5 \\
11 4 \\
12 3 \\
13 3 \\
14 2 \\
15 2 \\
16 2 \\
17 1 \\
18 1 \\
19 1 \\
20 1 \\
21 1 \\
};
\end{axis}
\end{tikzpicture}
\caption{$q=7$}         
\label{fig:code_b}
\end{subfigure}\hfill
\vspace{0.2cm}
\pgfplotslegendfromname{mylegend}
\caption{Bounds on the code size for $\6n=(7,\,7)$, $\bm{\lambda} = (1,\,2)$.}
\label{fig:more_comparison}
\end{figure}

\section{Code Construction}\label{sec:new_codes}

This section considers the case $d=5$ and $\bm{\lambda} = (1,\,2)$.
That is, correctable error patterns $(\6e_1, \6e_2)$ that satisfy $\wtH(\6e_1)\leq 2$ and $\wtH(\6e_2)=0$, or $\wtH(\6e_1)= 0$ and $\wtH(\6e_2)\leq1$.
In this setting, binary Hamming codes \cite{moon2018weighted} and binary generalized Goppa codes \cite{bezzateev2013class} are perfect for particular choices of block lengths.
We propose a different, simple code construction that works over an arbitrary finite field. The construction is optimal
for various field sizes and block lengths.

\begin{construction}\label{const:simple}
Let $\bm{\lambda}=(1,2)$, and $\mathcal{C} = \ker(\6H)$ for
\[
\6H = 
\begin{pmatrix}
    \6H_1 & \6H_2\\
    \6H_3 & \6 0\\
\end{pmatrix},
\]
with $\6H_1\in\mathbb{F}_q^{r_1 \times n_1}$ , $\6H_2\in\mathbb{F}_q^{r_1 \times n_2}$ and $\6H_3\in\mathbb{F}_q^{r_2 \times n_1}$. 
We pick 
\begin{itemize}
    \item $\6H_2, \6H_3$ as parity-check matrices of codes with minimum Hamming distance three,
    \item and $(\6H_1^\top, \6H_3^\top)^\top$ as a parity-check matrix of a code with minimum Hamming distance five. 
\end{itemize}
Then,~$\mathcal{C}$ has minimum \wH distance $d(\mathcal{C}) = 5$.
\end{construction}

\begin{proof}
To prove that Construction~\ref{const:simple} provides error-correction capability $\tau(\mathcal{C})\geq 2$, we provide a simple decoding algorithm.
Let $\6r =(\6r_1, \6r_2)= \6c + \6e = (\6c_1+\6e_1, \6c_2+\6e_2)$ with $\6c\in\mathcal{C}$ and $\text{wt}(\6e)\leq 2$.
To decode $\6r$, one calculates the syndrome $\6s_3$ of $\6r_1$ with respect to $\6H_3$.
Since $\6H_3$ enables us to detect two errors, $\6s_3=\60$ if and only if $\6e_1 = \60$.
If $\6s_3=\60$, we get $\6e_1 = \60$ and $\wtH(\6e_2) \leq 1$ and can thus use $\6H_2$ to correct.
If $\6s_3 \neq \60$, then $\wtH(\6e_1) \leq 2$ and $\6e_2 = \60$. Thus, $(\6H_1^\top, \6H_3^\top)^\top$ can correct the error $\6e_1$.
\end{proof}

Unlike for the general case (see Section~\ref{sec:new_weight}) the error-correction capability of the codes provided by Construction~\ref{const:simple} is precisely characterized by the minimum distance:
let $\6c_2\in\ker(\6H_2)$  with $\wtH(\6c_2) = 3$.
Then, $(\60,\,\6c_2)\in\mathcal{C}$ implies
\[
\tau(\mathcal{C}) \leq \min_{\6r\in\mathbb{F}_q^{n_2}} \max\{\text{wt}(\6r), \text{wt}(\6c_2-\6r)\} - 1 = 2.
\]

\begin{theorem}
For $q=2$ and $n_1 = n_2 = 2^m-1$, Construction~\ref{const:simple} achieves the highest dimension possible according to the Hamming-like bound (Thm.~\ref{theo:Hamming}), that is
\[
k^* = 2(2^m-m-1)\quad\text{or}\quad n-k^* = 2 m.
\]
\end{theorem}

\begin{proof}
For $\bm{\lambda}= (1,2)$, we get  $\lvert \ball_2(n,2, \bm{\lambda})\rvert = 2^{m-1}  (2^m+1).$
According to the Hamming bound, the minimum redundancy required to correct all errors in $\ball_2(n,2, \bm{\lambda})$ is
\[
\left\lceil \log_2(\lvert \ball_2(n,2, \bm{\lambda})\rvert) \right\rceil = \left\lceil \log_2(2^{m-1}) + \log_2(2^{m}+1)\right\rceil = 2 m,
\]
i.e., no code with $k>k^*$ can have $d=5$.
Next, we show that Construction \ref{const:simple} requires no more than $2 m$ bits of redundancy.
We pick $\6H_2, \6H_3 \in \mathbb{F}_2^{m\times n_\ell}$ as  parity-check matrices of a Hamming code.
Then, $\6H_1\in \mathbb{F}_2^{m\times n_\ell}$ can be picked such that it extends $\6H_3$ to a parity-check matrix of a double-error-correcting BCH code \cite[Chapter~3]{macwilliams1977theory}. 
\end{proof}
Further, using Construction \ref{const:simple}, we can build MWHD codes, which are \emph{not necessarily} MDS codes.
\begin{theorem}
For $q\geq \max\{n_1, n_2\}$ and $n_1 \geq 5$, one can use Construction \ref{const:simple} to achieve the Singleton-like bound, that is 
\[
k^* = n_1 + n_2 - 4 \quad\text{or}\quad n-k^* = 4.
\]
\end{theorem}
\begin{proof}
Since $n_1\lambda_1 \geq d$, the Singleton-like bound in \Cref{theo:singleton} implies $\lvert\mathcal{C}\rvert \leq q^{n-(d-1)/\lambda_1} = q^{n-4},$ i.e., any code with minimum \wH distance $5$ requires $n-k \geq 4$.
Next, we show that Construction \ref{const:simple} requires no more than $4$~redundancy symbols.
We pick $\6H_2, \6H_3 \in \mathbb{F}_q^{2\times n_\ell}$ as parity-check matrices of a single-error-correcting MDS code.
Then, $\6H_1\in \mathbb{F}_q^{2\times n_1}$ can be picked such that it extends $\6H_3$ to a parity-check matrix of a double-error-correcting MDS code. 
\end{proof}

\section{Conclusion}\label{sec:conc}

This paper studies the \wH metric, which is tailored to independent parallel $q$-ary symmetric channels.
For suitable scaling factors, minimum-distance decoding achieves optimal performance.
We observe that the \wH metric is not normal, i.e., there are codes for which the error-correction capability exceeds half the minimum distance.
We bound the error-correction capability of a code via its minimum distance, which we, in turn, bound by generalizing the Singleton, Plotkin, Hamming, Gilbert-Varshamov, and linear programming bounds.
A simple code construction with optimal minimum distance is proposed for specific parameters.

Finally, we want to point out that some applications allow parallel channels to be operated with different alphabets \cite{sidorenko2005polyalphabetic}; we delegate an extension of the \wH metric to polyalphabetic codes to future work.

\section*{Acknowledgment}
Violetta Weger is  supported by the European Union's Horizon 2020 research and innovation programme under the Marie Sk\l{}odowska-Curie grant agreement no. 899987.
Sebastian Bitzer  acknowledges the financial support by the Federal Ministry of Education and
Research of Germany in the program of “Souverän. Digital. Vernetzt.”. Joint project 6G-life, project
identification number: 16KISK002.
Alberto Ravagnani is supported by the Dutch Research
Council via grants VI.Vidi.203.045, OCENW.KLEIN.539, and by the Royal Academy of
Arts and Sciences of the Netherlands.
\bibliographystyle{plain}
\bibliography{references}

\end{document}